\documentclass[11pt]{article}

\usepackage[margin=2.5cm]{geometry}
\usepackage{authblk}

\usepackage{amssymb,amsmath,amsthm,amsfonts,thmtools}
\usepackage{color}
\usepackage{algorithm}
\usepackage{yfonts}
\usepackage{mathrsfs}

\usepackage{graphicx}
\usepackage{xspace}

\usepackage{times}
\usepackage{verbatim}
\usepackage{enumitem}
\usepackage{graphicx,wrapfig,lipsum}
\usepackage[font=small]{caption}
\usepackage{subcaption}
\usepackage[english]{babel}
\usepackage[utf8]{inputenc}

\usepackage[noend]{algpseudocode}
\usepackage{url}
\usepackage[all]{xy}
\usepackage{rotating}
\usepackage{ifpdf}
\usepackage{tcolorbox}
\usepackage{lipsum}

\usepackage{mdframed}             
\usepackage{xifthen} 
\usepackage{mathtools}  

\DeclareMathAlphabet\mathbfcal{OMS}{cmsy}{b}{n}

\usepackage{comment}
\usepackage[T1]{fontenc}

\usepackage{lineno}
\linenumbers

\usepackage[hidelinks]{hyperref} 

\usepackage[section]{placeins}
\usepackage{float}
\usepackage[title]{appendix}
\usepackage{graphicx}
\usepackage{tikz}
\usetikzlibrary{automata,positioning}
\usepackage{dsfont}
\usepackage{xspace}
\usepackage{enumitem}
\setdescription{font=\normalfont}

\newcommand{\barr}{{\bar r}}

\newcommand{\barS}{{\bar S}}

\newcommand{\tildeO}{{\tilde{O}}}

\colorlet{darkgreen}{green!45!black}

\newcommand{\Disperse}{{\textsc{Disperse}}}
\newcommand{\QuasiGossip}{{\textsc{QuasiGossip}}}

\newtheorem{theorem}{Theorem}

\newcommand{\ignore}[1]{}

\newcommand{\margincomment}[2]%
{\marginpar{\footnotesize\raggedright {\color{red}#1}: #2}}

\newcommand{\etal}{et~al.\ }

\newcommand{\myparagraph}[1]{{\medskip\noindent\textbf{#1}}}

\newcommand{\braced}[1]{{ \left\{ {#1} \right\} }}

\newcommand{\half}{\textstyle{\frac{1}{2}}}
\newcommand{\loglog}{\log\!\log}

\newcommand{\kNpair}[2]{$(#1,#2)$}
\newcommand{\kqNtriple}[3]{$(#1, #2, #3)$}

\newcommand{\ListLengths}{\setlength{\itemsep}{0ex}\setlength{\topsep}{1ex}\setlength{\partopsep}{0ex}}

\graphicspath{{./Figures/}}
\title{On Permutation Selectors and their Applications in Ad-Hoc Radio Networks Protocols\thanks{Research supported by NSF grant CCF-2153723.}}

\author[1]{Jordan Kuschner}
\author[1]{Yugarshi Shashwat}
\author[1]{Sarthak Yadav}
\author[1]{Marek Chrobak}

\affil[1]{University of California at Riverside}

\begin{document}
\nolinenumbers

\maketitle

\begin{abstract}
Selective families of sets, or selectors, are combinatorial tools used to ``isolate'' individual members of sets 
from some set family. Given a set $X$ and an element $x\in X$,
to isolate $x$ from $X$, at least one of the sets in the selector must intersect $X$ on exactly $x$.  
We study \emph{\kNpair{k}{N}-permutation selectors} which have the property that they can isolate each element of each $k$-element
subset of $\braced{0,1,...,N-1}$ in each possible order. These selectors can be used in protocols for
ad-hoc radio networks to more efficiently disseminate information along multiple hops.
In 2004, Gasieniec, Radzik and Xin gave a construction of a \kNpair{k}{N}-permutation selector of size $O(k^2\log^3 N)$.
This paper improves this by providing a probabilistic construction of a \kNpair{k}{N}-permutation selector of size $O(k^2\log N)$.
Remarkably, this matches the asymptotic bound for standard strong \kNpair{k}{N}-selectors, that isolate each element of each set of size $k$, 
but with no restriction on the order.
We then show that the use of our \kNpair{k}{N}-permutation selector improves the best running time for
gossiping in ad-hoc radio networks by a poly-logarithmic factor. 
\end{abstract}


\section{Introduction}
\label{sec: introduction}


Selective families of sets, or selectors, are combinatorial tools used to ``isolate'' individual members of sets 
belonging to a given collection of sets. Given a set $X$ and some element $x\in X$,
to isolate $x$ from $X$, at least one of the sets in the selector must intersect $X$ on exactly $x$. 
Various types of selectors have been constructed in the literature and used in applications ranging
from group testing to coding, bioinformatics, multiple-access channel protocols, and information dissemination in radio networks. 

To illustrate this concept on a concrete example, consider the contention resolution problem in multiple-access channels (MACs)
without feedback: $N$ devices are connected to a shared broadcast channel (say, a radio frequency or ethernet), and
each has a unique identifier from the set $U = \braced{0,1,...,N-1}$. 
If exactly one device transmits in some time slot, then its message will be delivered through the channel to all
other devices. However, if two or more devices transmit simultaneously, then a collision on the broadcast channel occurs.
Channel collisions are indistinguishable from background noise, and the sender does not receive
any feedback about the fate of its transmission.

Suppose that some unknown set of $k$ devices initially activate and have messages to be transmitted.
(See Figure~\ref{fig: mac and in-neighbors}.)
We seek a protocol that would allow these $k$ active devices to transmit successfully, providing that the other devices remain idle.
Without any feedback, each protocol for this model is non-adaptive, that is, it
can be uniquely identified with the sequence of transmission sets $\barS = S_0, S_1, ..., S_{m-1}$, where each
$S_t$ is the set of devices allowed to transmit at time slots $t$, if active.
A trivial protocol would avoid collisions altogether by using singleton transmission sets $\braced{0}, \braced{1}, ..., \braced{N-1}$,
but this requires $m=N$ steps to complete $k$ transmissions, which seems wasteful if $k$ is small.

For an active device $x$ to transmit successfully at a slot $t$ of a transmission sequence $\barS$, $x$ must be the only
active device in set $S_t$. Therefore $\barS$ must satisfy the following property:
for any set $X\subseteq U$ of cardinality $k$ and any $x\in X$
there is $t\in\braced{0,1,...,m-1}$ such that $S_t\cap X = \braced{x}$. 
A family $\barS$ that satisfies this condition is called a \emph{strong \kNpair{k}{N}-selector}.
In other words, a strong \kNpair{k}{N}-selector isolates each element of each subset of $U$ of cardinality $k$.
Optimizing the time to complete all transmissions in the above MAC model is equivalent to
constructing a strong \kNpair{k}{N}-selector of minimum size $m$.

Very similar challenges arise in information dissemination protocols for ad-hoc radio networks.
A radio network is a directed graph whose nodes represent radio transmitters/receivers and edges
represent their transmission ranges. There are $n$ nodes, each assigned a unique label from  $U = \braced{0,1,...,N-1}$. 
When a node transmits, its message is sent to all its out-neighbors; 
however, the possibility of collisions at the out-neighbors can result in the loss of transmitted messages.
The ad-hoc property dictates that the nodes are oblivious to the network's topology at the beginning. 
In particular, in the scenario with a node $v$ having $k$ in-neighbors that attempt to transmit to $v$
the challenge is equivalent to the MAC contention resolution. 
( See Figure~\ref{fig: mac and in-neighbors}.)
A protocol that applies a strong \kNpair{k}{N}-selector will guarantee that in at most $m$ steps
all in-neighbors of $v$ will successfully transmit their messages to it.

\begin{figure}[t]
		\centering
		\begin{subfigure}[c]{0.48\textwidth}
			\centering
   		 \includegraphics[width=2.25in]{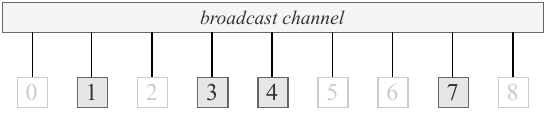}
		\end{subfigure}
	\begin{subfigure}[c]{0.48\textwidth}
		\centering
		    \includegraphics[width=2in]{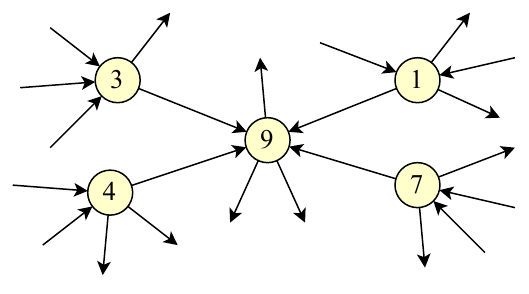}
	\end{subfigure}
		\caption{On the left, an illustration of MAC contention resolution, with
				$N=9$ and $k=4$. Active devices are marked with darker colors.
				On the right, a node $v$ in a radio network, with label $9$ and four in-neighbors.
				} \label{fig: mac and in-neighbors}
\end{figure}


Strong \kNpair{k}{N}-selectors have been well studied in the literature (see the discussion at the end of this section). 
It is known that
there are strong \kNpair{k}{N}-selectors of size $m = O(k^2\log N)$, beating the earlier-mentioned trivial bound of $N$ if $k$ is small.


\paragraph{Our contribution.}
We study a new type of selectors called \emph{\kNpair{k}{N}-permutation selectors}.
We earlier saw that a strong \kNpair{k}{N}-selector isolates all elements of each subset  $X$ of $U$ of size $k$.
The \kNpair{k}{N}-permutation selector guarantees an extra property that for each possible permutation of $X$, 
it isolates all $x \in X$ in the order of this permutation. 

The motivation comes from radio networks, where, unlike in the MAC contention resolution problem, a
message may need to travel along a path of multiple nodes. At each node $v$ along the path, the in-neighbors of $v$ need to overcome contention to successfully transmit to $v$. 
Consider such a path $P = v_0v_1...v_s$, and suppose that the in-neighborhood $X$ of $P$ (the set
of nodes with edges going to $P$) has at most $k$ nodes. (See Figure~\ref{fig: path neighborhood} for illustration.)
An $s$-fold repetition of a strong \kNpair{k}{N}-selector will deliver a message from $v_0$ to $v_s$ in time 
$O(sk^2\log N)  = O(k^3\log N)$.  A \kNpair{k}{N}-permutation selector of size $m$ can achieve this in time $m$, because
it guarantees that in $m$ steps the nodes $v_0,v_1,...,v_{s-1}$ will successfully transmit, one by one, \emph{in this order}.


\begin{figure}[t]
	\centering
    \includegraphics[width=3in]{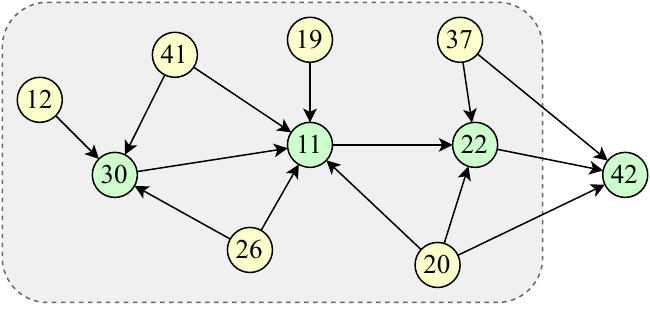}
	\caption{An example of a path and its in-neighborhood, with nodes identified by their labels.
		The path is $P = 30,11,22,42$ and its neighborhood (shaded) has $9$ nodes.
	}
	\label{fig: path neighborhood}
\end{figure}


The concept of \kNpair{k}{N}-permutation selectors is implicit in the work by Gasieniec, Radzik and Xin~\cite{Gasieniec_etal_det_gossip_04}.
They show that a \kNpair{k}{N}-permutation selector can be constructed by interleaving other known types of selectors%
\footnote{We remark that in~\cite{Gasieniec_etal_det_gossip_04} the authors use a different style for specifying the
parameters of selectors. The notation in our paper follows the convention from~\cite{DeBonis_etal_selectors}.}.
Their construction is called a \emph{path selector} and it has size 
$O(k^2\log N\log^2 k)$. Thus, assuming that $N$ is polynomial in the network size $n$, a protocol 
based on their path selector will deliver a message along a path with in-neighborhood of size at most $k$ in time 
$O(k^2 \log^3 n)$. They used this idea to give an $O(n^{4/3} \log^{10/3} n)$-time protocol for gossiping
in ad-hoc radio networks. This is the best currently known upper bound for this problem.

The main result of our paper is an improved upper bound on the size of  \kNpair{k}{N}-permutation selectors. Using a
probabilistic construction, we prove the following theorem in Section~\ref{sec: permutation selectors}.


\begin{theorem}\label{thm: main result}
There is a \kNpair{k}{N}-permutation selector of size $m = O(k^2\log N)$.
\end{theorem}


This bound matches the best bound on strong \kNpair{k}{N}-selectors.
This seems surprising, as \kNpair{k}{N}-permutation selectors offer more capability: they can isolate any
$k$ elements in any given order, while strong selectors will only do so in some unknown order.
Theorem~\ref{thm: main result} leads to a poly-logarithmic improvement of the time complexity of gossiping in ad-hoc 
radio networks. A further minor improvement can be achieved by using 
a faster procedure for broadcasting~\cite{Czumaj_Davies_faster_16} inside the gossiping protocol.
This leads to the following result:


\begin{theorem}\label{thm: gossiping result}
If $N$ is polynomial in $n$, then the gossiping problem in ad-hoc radio networks can be solved in
time $O(n^{4/3}\log n (\loglog n)^{2/3})$.
\end{theorem}


\smallskip

In Appendix~\ref{sec: nkq permutation selectors}, we also consider an even more general
structure named \emph{\kqNtriple{k}{q}{N}-permutation selectors}. These are defined by the following property: 
for each permutation of each $k$-element set $X$, the selector isolates some $q$-elements of $X$ in the order of this permutation.
This naturally extends the concept of \kqNtriple{k}{q}{N}-selectors, which isolate $q$ elements of a
$k$-element set without restricting the order. 
(See the discussion below.)
Extending the proof of Theorem~\ref{thm: main result}, we show that there exists
a \kqNtriple{k}{q}{N}-permutation selector of size $O(kq\log N)$.


\paragraph{Related work.}
Combinatorial structures closely related to selectors have been studied in different settings
and under different terminology, with connections between these concepts not always obvious. 
Some of these structures are equivalent to strong \kNpair{k}{N}-selectors, while others are related to the concept 
of \emph{weak \kNpair{k}{N}-selectors}, which only isolate any one element of each $k$-element set.
Examples include superimposed codes used in information retrieval~\cite{kautz_singleton_1964},
cover-free set families~\cite{Erdos_etal_families_85}, 
as well as protocols for non-adaptive group testing~\cite{hwang_sos_group_testing_1987}
and for MAC contention resolution~\cite{komlos_greenberg_non-adaptive_mac_1985}.
The use of selectors in protocols for ad-hoc radio networks was initiated in~\cite{basagni_etal_deterministic_broadcast_1999,chlebus_etal_deterministic_broadcasting_2002},
and various forms of selectors have been used in essentially all deterministic
protocols for information dissemination in such networks.

One classical example of applications of selectors outside of networking is group testing,
a method used in fields such as medical diagnostics, quality control, or bioinformatics. 
The goal in group testing is to efficiently identify  individuals who test positive for a specific trait, such as a disease. 
Instead of testing these individuals separately, group testing works by pooling individuals 
into groups for collective testing. With this approach, selectors can be employed to minimize the number of required tests.
For more thorough discussion of applications of selectors, see~\cite{DeBonis_etal_selectors} and the references therein.

For strong \kNpair{k}{N}-selectors, the upper bound of $O(k^2\log N)$ can
be established by a probablistic construction~\cite{kautz_singleton_1964,Erdos_etal_families_85,komlos_greenberg_non-adaptive_mac_1985}.
An explicit construction matching this bound can be found in~\cite{porat_rothschild_explicit_group_testing_2011}. 
Note that such bounds are of interest only for $k = O(\sqrt{N/\log N})$, because
for larger values of $k$ a trivial construction using $N$ singletons is better.
This is essentially tight, given a lower bound of $\Omega(k^2\log^{-1}k \log N)$ (for $k = O(\sqrt{N})$) 
established in~\cite{Clementi_etal_distributed_02,dyachkov_rykov_bounds_codes_1982,chaurdhuri_radhakrishnan_deterministic_1996}.

De Bonis~{\etal}~\cite{DeBonis_etal_selectors} introduced a more general model of selectors
called \emph{\kqNtriple{k}{q}{N}-selectors}. A \kqNtriple{k}{q}{N}-selector has the property that it
can isolate at least $q$ elements from each $k$-element subset of $U$. 
A strong \kNpair{k}{N}-selector is a special case when $q=k$, and a weak \kNpair{k}{N}-selector is a special case when $q = 1$.
They proved that there are \kqNtriple{k}{q}{N}-selectors of size $O(k^2/(k - q + 1)\log N)$.


\section{Permutation Selectors}
\label{sec: permutation selectors}



The objective in this section is to prove Theorem~\ref{thm: main result}, namely to show that
there is a  \kNpair{k}{N}-permutation selector of size $m = O(k^2\log N)$.

We start with a more formal definition of our permutation selectors. Let $\barS = S_0, S_1, ... , S_{m-1}$ be a sequence of subsets of $U$.
We call $\barS$ a \emph{\kNpair{k}{N}-permutation selector} if it has the following property:
\begin{description}
\item{(PS)} For every  $X\subseteq U$ with $|X| = k$ and for each permutation $\pi= x_1, x_2, ..., x_k$ of $X$, 
	there exists an increasing sequence of indices $0  \le i_1 < i_2 < ... < i_k\le m-1$ 
	such that $S_{i_l}$ isolates $x_{l}$ from $X$ for each $l= 1,2, . . . , k$.
\end{description}

If $\barS$ satisfies this property for a set $X$ and a permutation $\pi$, 
we will say that it \emph{isolates $X$ in order $\pi$}, or simply that $\barS$ \emph{isolates $\pi$}.
Thus $\barS$ is a \kNpair{k}{N}-permutation selector if it isolates each $k$-permutation of $U$.

\medskip

The proof idea is to choose each set $S_i$ randomly, by having each label in $U$, independently, add itself to $S_i$ with probability $\frac{1}{k}$. 
This way, for any fixed set $X$ of cardinality $k$, each $S_i$ will isolate some element of $X$ with probability
$\gamma = (1-1/k)^{k-1}$, so $\gamma$ is in the range $\frac{1}{e} < \gamma \le \half$. 
We then show that the probability that $X$ is not isolated in order $\pi$ is
exponentially small with respect to $m/k$. This, using the union bound, will give us that the probability 
that $\barS$ is \emph{not} a \kNpair{k}{N}-permutation selector is less than $1$ if $m = \Theta(k^2\log N)$, with a sufficiently large constant hidden behind the big-Theta. 
Therefore, some sequence $\barS$ of length $m = O(k^2\log N)$ is a \kNpair{k}{N}-permutation selector.

We now proceed with the details.
As in the first part of the proof we fix $X$ and $\pi$, we can as well assume that $X = \braced{0,1,...,k-1}$ and
that the desired permutation of $X$ is $\pi = 0,1,...,k-1$. Our first goal is to estimate the probability
that $\barS$ does not isolate $X$ in order $\pi$.

To this end, we start by considering an auxiliary problem, which is basically a variant of coupon collection:
Suppose that we generate uniformly a random sequence $\barr = r_0,r_1,...,r_{\ell-1}$ of elements of $\braced{0,1,...,k-1}$, where $\ell \ge k \ge 2$.
We want to compute the probability $p_{\ell,k}$ of the event ``$\barr$ does not contain $\pi$ as a subsequence''.
To compute $p_{\ell,k}$, we reason as follows. For a given $j\in\braced{0,1,...,k-1}$,
if $\barr$ contains $0,1,..,j-1$ as a sub-sequence,
associate with $\barr$ the unique \emph{lexicographically-first appearance} of $0,1,..,j-1$,
namely the increasing sequence $i_0,i_1,...,i_{j-1}$ of positions in $\barr$,
where $i_0$ is the position of the first $0$, $i_1$ is the first position of $1$ after $i_0$,
and so on. The number of $\barr$'s 
that contain $0,1,2,...,j-1$ but not $0,1,2,...,j$ can be computed by multiplying
the number  of choices, $\binom{\ell}{j}$, for the lexicographically-first appearance of $0,1,...,j-1$,
and the number of ways for filling the remaining $\ell-j$ positions, which is $(k-1)^{\ell-j}$.
This yields the formula for $p_{\ell,k}$, given below, for which we then derive an upper bound estimate.
\begin{align}
p_{\ell,k} \;&=\;  \frac{1}{k^\ell}\cdot \sum_{j=0}^{k-1} \binom{\ell}{j} (k-1)^{\ell-j}
\nonumber \\
		&= \; (1- 1/k)^\ell \cdot \sum_{j=0}^{k-1} \binom{\ell}{j} (k-1)^{-j}
\nonumber \\
		&< \; e^{-\ell/k} \cdot \sum_{j=0}^{k-1} \left(\frac{\ell}{k-1}\right)^{j}
		\; \le \; e^{-\ell/k} ({2\ell}/{k})^k,
		\label{eqn: estimate on pellk}
\end{align}
where the last step follows from $\ell/(k-1) \le 2\ell/k$ 
and $\sum_{j=0}^{k-1} (2\ell/k)^{j} < (2\ell/k)^k/(2\ell/k-1) \le (2\ell/k)^k$, as $\ell\ge k$.

Next, still with $X$ and $\pi$ fixed as above, we estimate the probability that $\barS$ does not isolate $X$ in order $\pi$. 
Let $h$ be the random variable representing the number of sets $S_i$ that isolate some element of $X$ 
(with repetitions counted). As stated earlier, the probability that some element of $X$ is isolated by a given $S_i$ is 
$\gamma \in (\frac{1}{e} , \half]$. So $h$ is a binomial random variable with success probability $\gamma$ and mean $\mu = \gamma m$.
Therefore, letting $\delta = 1 - 1/(4\gamma)$ and using the Chernoff bound for the lower tail of $h$'s distribution, we get
\begin{equation}
	\Pr[ h \le m/4] \;=\; \Pr[ h \le (1-\delta) \mu] 
	\;\le\; e^{-\delta^2 \mu/2}	
	\;=\; e^{-\delta^2 \gamma m /2}
	\;=\; \alpha^{m} ,
	\label{eqn: reduction to coupons}
\end{equation}
where $\alpha = e^{-\delta^2 \gamma / 2} < 1$. 

Consider now the sequence consisting of elements of $X$ (with repetitions) that are isolated by $\barS$, in order in which
they are isolated.
Let $E_\ell$ be the event that the first $\min(\ell,h)$ elements of this sequence do not
contain our permutation $\pi$ as a subsequence.  Then $\Pr[E_{\ell} | h \ge \ell] = p_{\ell,k}$, so
\begin{align}
	\Pr[\barS\; \textrm{does not isolate $X$ in order $\pi$} ]\;&\le\;
					\Pr[h \le m/4] + \Pr[\,E_{m/4} \,|\, h \ge m/4 \,] \nonumber
					\\
		&\le\; \alpha^m + p_{m/4,k} \nonumber
					\\
		&\le\; \alpha^m + e^{-m/4k} ({m}/{2k})^k \nonumber
		\\
		&\le\; \beta^{m/k} ({m}/{k})^k  \,, \label{eq: prob not isolating X}
\end{align}
for $\beta = \max(\alpha, e^{-1/4})$. (We can assume that $m \ge 2k$.)

Let $m = ck^2\log N$, for some constant $c$ that will be specified later.
To upper bound the probability that $\barS$ is not a \kNpair{k}{N}-permutation selector, we apply the union bound. 
We have $\binom{N}{k}$ choices of $k$-element sets $X$, and each can be permuted in $k!$ ways.
So, using the bound~\eqref{eq: prob not isolating X}, we obtain

\begin{align}
\Pr[\barS\; \textrm{is not a \kNpair{k}{N}-permutation selector}] 
		\;& \le \;  \textstyle \binom{N}{k}\cdot k!  \cdot [ \beta^{m/k} ({m}/{k})^k ]
		\nonumber
	\\
	&\le \; \textstyle \binom{N}{k}\cdot k!  \cdot  \beta^{ck  \log N} \cdot (ck\log N)^{k}  
	\nonumber
	\\
	&\le \; N^{k}\cdot  N^k  \cdot \beta^{ck\log N} \cdot c^{k \log N} \cdot N^k \cdot N^k
	\nonumber
	\\
	&= \; N^{4k} \cdot  ( c \beta^c)^{k  \log N} \;<\; 1  \,, 
	\label{eqn: union bound}
\end{align}
where the last inequality holds if $c$ is large enough so that $c\beta^c < \frac{1}{16}$.
This proves that for this $c$ and $m = ck^2\log N$ there exists a \kNpair{k}{N}-permutation selector $\barS$ of length $m$.


\section{Gossiping in Ad-Hoc Radio Networks}
\label{sec: gossiping in ad-hoc radio networks}


In this section we prove Theorem~\ref{thm: gossiping result}. We start by giving a formal
description of the ad-hoc radio network model. The gossiping protocol we use is
essentially identical to the one in~\cite{Gasieniec_etal_det_gossip_04}, but we include 
its high-level description and sketch the analysis, for the sake of completeness.

\myparagraph{Ad-hoc radio network model.}
A radio network can be naturally modeled as a directed graph $G = (V, E)$ whose nodes represent processing elements
equipped with radio transmitters/receivers.
Each node has a unique label from the set $U = \braced{0,1, \ldots, N-1}$. The edges
represent the nodes' transmission ranges, that is $(u, v) \in E$ iff $v$ is in the range of the transmissions from $u$.
If $(u,v)\in E$ then $v$ is called an \textit{out-neighbor} of $u$ and  $u$ is an \textit{in-neighbor} of $v$. 

Initially, all nodes know only their own label and the upper bound $N$ on the number of labels. They do not
have any information about the network's topology. The time is discrete, divided into equal-length time steps.
If a node $u$ transmits a message, this message is sent to all its out-neighbors at the same time step. 
If $v$ is one of these out-neighbors, and $u$ is the only in-neighbor of $v$ transmitting at this step,
then $v$ will receive $u$'s message. But if some other in-neighbor of $v$ transmits at the same step, a \emph{collision} occurs.  
The model does not assume any collision detection capability, so $u$ will not receive any collision
notification and $v$ will not know that $u$ transmitted. There are no restrictions on message size or local computation.

The two most basic information dissemination primitives in this model are broadcasting and gossiping.
In broadcasting (or one-to-all communication) the goal is to deliver a message from a designated \emph{source} node to all
other nodes. In gossiping (or all-to-all communication) each node starts with its own piece of information that we call a \emph{rumor},
and the rumors from each node must be delivered to all other
nodes. For these problems to be well defined, $G$ must satisfy appropriate connectivity assumptions:
in broadcasting all nodes must be reachable from the source node, and in gossiping the network must be strongly connected.


\myparagraph{The gossiping protocol.} 
We describe the protocol under the assumption that $N=n$, and later we will show how to extend it to
the general case, when $N$ is polynomial in $n$. With this assumption,
a trivial gossiping protocol would broadcast the rumors from nodes labeled $0,1,\ldots,n-1$
one by one. Denoting by $B(n)$ the running time of a broadcasting protocol, this would take time $O(nB(n))$.
To speed this up, the gossiping protocols in~\cite{Chrobak_etal_fast_02,Liu_Prabhakaran_randomized_02,Gasieniec_etal_det_gossip_04}
work by grouping rumors in some nodes so that these nodes can broadcast the collected rumors in a single message.

To reduce the number of such broadcasts, the idea is to broadcast from nodes that have many rumors.
We call a rumor \emph{active} if it has not been broadcast yet, and a node is \emph{active} if its rumor is active
and \emph{dormant} otherwise.
We use procedure $\Disperse(\mu)$, which repeatedly chooses a node with at least $\mu$ active
rumors and then broadcasts from this node. Choosing such a node can be accomplished with
broadcasting and binary search~\cite{Chrobak_etal_fast_02}, and the number of chosen nodes will
be $O(n/\mu)$, so the total running time of $\Disperse(\mu)$ is $O((n/\mu)B(n)\log n)$.

One clever observation in~\cite{Gasieniec_etal_det_gossip_04} is that it is sufficient to
give a protocol for a task called \emph{quasi-gossiping}, where each node needs to either become
dormant itself or have its rumor delivered to a dormant node. This is because a single repetition of
the transmission sequence from the quasi-gossiping protocol will in fact complete full gossiping.
The pseudo-code for the quasi-gossiping protocol
is given in Algorithm~\ref{alg: quasigossip}. It uses a parameter $\kappa$ whose value will be determined later.


\algdef{SE}[REPEATN]{RepeatTimes}{EndRepeatTimes}[1]{\algorithmicrepeat\ #1 \textbf{times}}{}
\algnotext{EndRepeatTimes}

\begin{algorithm}[t]
\caption{\QuasiGossip} \label{alg: quasigossip}
\begin{algorithmic}[1]
	\For{$v=0,1,...,n-1$}
   	 	\State {transmit from node $v$}
	\EndFor
    \State{$\Disperse(\kappa)$}
    \RepeatTimes{$\log \kappa + 1$}
    	\State {the active nodes transmit according to a \kNpair{\kappa}{n}-permutation selector}
   	 	\State {$\Disperse(\kappa/2)$}
	\EndRepeatTimes
\end{algorithmic}
\end{algorithm}


The correctness of Protocol~$\QuasiGossip$ is justified by focussing on \emph{active paths}, which are
paths consisting only of active nodes. For any active path, define its \emph{active in-neighborhood} to be the
set of active in-neighbors of the nodes on this path.
Let $\ell$ be the largest number such that each active path with $\ell$ nodes has
its active in-neighborhood size smaller than $\kappa$. After the execution of $\Disperse(\kappa)$, each
node has fewer than $\kappa$ active in-neighbors, so at this point we have $\ell\ge 1$.
Each iteration of the \textbf{repeat} loop at least doubles the value of $\ell$.
Since all nodes on an active path (except possibly the last) belong to the active in-neighborhood,
the value of $\ell$ cannot exceed $\kappa$, so after $\log \kappa$ iterations
\emph{all} active paths will have active in-neighborhoods bounded by $\kappa$, and then
the last iteration will complete the quasi-gossiping task.

In line~5 we use the \kNpair{\kappa}{n}-permutation selector from Theorem~\ref{thm: main result},
so the total cost of these selectors will be  $O(\kappa^2\log^2 n)$. 
The cost of all calls to $\Disperse()$ is $O((n/\kappa+\log n)B(n)\log n)$.
Thus the overall running time of Protocol~$\QuasiGossip$ is asymptotically bounded by:
\begin{equation*}
(n/\kappa + \log n)B(n)\log n + \kappa^2\log^2 n
\end{equation*}
Letting $\kappa = (nB(n)/\log n)^{1/3}$, and 
using the bound $B(n) = O(n \log n \loglog n)$ from~\cite{marco_distributed_broadcast_2010,Czumaj_Davies_faster_16}
on the complexity of broadcasting, we obtain a gossiping protocol 
with running time $O(n^{4/3}\log^2 n (\loglog n)^{2/3})$. 


\smallskip

To extend Protocol~$\QuasiGossip$ to the case when $N$ is polynomial in $n$, we first replace $n$ by $N$ in all
invocations of selectors. After this, the only problematic part of Protocol~$\QuasiGossip$ is the \textbf{for} loop in Lines~1-2 that
reduces the active in-neighborhoods of the individual nodes. Instead of this loop, 
the protocol in~\cite{Gasieniec_etal_det_gossip_04} uses \kqNtriple{s}{s/4}{N}-selectors for geometrically decreasing values of $s$,
combined with operations $\Disperse(s/4)$, to gradually
reduce these in-neighborhoods.
(This is similar to the method in~\cite{komlos_greenberg_non-adaptive_mac_1985}.)
The running time of this process
amortizes to $O((n/\kappa)B(n)\log n)$ --- same time as for the case of small labels.
So the overall running time also remains the same, completing the proof of Theorem~\ref{thm: gossiping result}.


\bibliographystyle{plain}
\bibliography{00_permutation_selectors_paper.bib}


\vfill\newpage

\appendix

\section{\kqNtriple{k}{q}{N}-Permutation Selectors}
\label{sec: nkq permutation selectors}



In Section ~\ref{sec: permutation selectors} we defined $(N,k)$-permutation selectors, a family of sets 
that isolates all elements of any $k$-permutation $\pi$ of the label set 
$U= \braced{0, 1, \ldots, N-1}$, meaning that all elements in $\pi$ are isolated in their listed order in $\pi$.

We now generalize this concept analogously to the way \kqNtriple{k}{q}{N}-selectors generalize standard selectors
by requiring that some $q$ elements of each $k$-permutation $\pi$ are isolated in the order of $\pi$.
Formally, let $\barS = S_0, S_1, \ldots, S_{m-1}$  be a sequence of subsets of $U$.
$\barS$ is called a \emph{\kqNtriple{k}{q}{N}-permutation selector} if the following property holds:
\begin{description}
\item{(PS')} For each $X\subseteq U$ with $|X|= k$, and for each permutation $\pi = x_1,x_2,...,x_k$ of $X$,
there are increasing sequences of indices $0 \leq i_1< i_2 < ... < i_{q} \leq m-1$ and $d_1 < d_2 < ... < d_q$
such that, for all $l = 1,2,...,q$, set $S_{i_l}$ isolates $x_{d_l}$ from $X$.
\end{description}
The remainder of this section provides a probabilistic construction of a \kqNtriple{k}{q}{N}-permutation selector, proving the following theorem.


\begin{theorem}\label{thm: Nkq permutation selector construction}
For any $k \in\braced{2,3,..., N}$, there exists a \kqNtriple{k}{q}{N}-permutation selector $\barS = S_0, S_1, \cdots, S_{m-1}$ of size $m = O(kq\log N)$.
\end{theorem}

\begin{proof}
As in the proof of Theorem~\ref{thm: main result}, we use a probabilistic argument.
The construction is the same: for each $j = 0,1,...,m-1$, we let $S_j$ be a random subset of $U$ obtained by
each element $x\in U$ adding itself to $S_j$, independently, with probability $\frac{1}{k}$. 
We then need to prove that if $c$ is a sufficiently large constant and $m = c \cdot qk\log N$, then
\begin{equation}
\Pr[\barS\; \textrm{is not a \kqNtriple{k}{q}{N}-permutation selector}] \;<\; 1.
\label{eqn: permutation Nkq-selector probability}
\end{equation}
The proof's high-level strategy is the same as in the proof of Theorem~\ref{thm: main result}.
We use the Chernoff bound to reduce the problem to a version of the coupon collection problem. 
We consider the following variant of the coupon collection problem:
For a random sequence $\barr = r_0,r_1,...,r_{\ell-1}$ of coupons from $\braced{0,1,...,k-1}$, where $\ell \ge k \ge 2$,
let $p'_{\ell,k,q}$ be the probability of the 
event ``$\barr$ does not contain an increasing subsequence of length $q$''. 
To show~\eqref{eqn: permutation Nkq-selector probability}, it is sufficient to prove that
\begin{equation}
p'_{\ell,k,q} \;\le\;  \gamma^{-\ell/q} ({2\ell}/{q})^q\,,
\label{eqn: Nkq-subsequence probability}
\end{equation}
for some constant $\gamma\in (0,1)$. Indeed, this is analogous to~\eqref{eqn: estimate on pellk}.
We can then use the Chernoff-based bound~\eqref{eqn: reduction to coupons}, to
obtain a bound of $\delta^{m/q}(m/q)^q$, for some $0 < \delta < 1$, 
on the probability that $\barS$ does not isolate $q$ elements of $\pi$ in order,
analogously to~\eqref{eq: prob not isolating X}.
For the union-bound estimate, we then take $m = ckq\log N$, with large enough $c$.
Since then $\delta^{m/q} = \delta^{ck\log N}$, 
the derivation for the union bound~\eqref{eqn: union bound} will be essentially the
same, with $\delta$ instead of $\beta$.

\smallskip

It remains to estimate $p'_{\ell,k,q}$.  We reason as follows. 
First, to avoid clutter in the calculations below, we will assume that $q$ is a divisor of $k$. 
(If it is not, the values of $k/q$ in the calculations below need to be rounded up or down.
This does not affect our asymptotic estimate.)

We now consider a special type of increasing sub-sequences that we refer to as \emph{$q$-jump sub-sequences}.
Partition the set of coupons $\braced{0,1,...,k-1}$ into $q$ equal size blocks $B_0, B_1,...,B_{q-1}$,
each having $k/q$ consecutive coupons. That is,
$B_h = \braced{h\frac{k}{q}, h\frac{k}{q}+1, ..., (h+1) \frac{k}{q}-1}$, for $h = 0,1,...,q-1$.
For $j\le q$,  a jump sub-sequence is a sequence of $j$ coupons $a_0,a_1,...,a_{j-1}$
such that $a_h \in B_h$ for each $h = 0,...,j-1$.

Define now $p''_{\ell,k,q}$ to be the probability of the event ``$\barr$ does not contain a jump sub-sequence of length $q$''.
Since $p''_{\ell,k,q} \ge p'_{\ell,k,q}$, it is sufficient to prove that the
same inequality~\eqref{eqn: Nkq-subsequence probability} holds for $p''_{\ell,k,q}$ instead of $p'_{\ell,k,q}$.
We do this by refining the argument in the proof of Theorem~\ref{thm: main result}.
For a given $j\in\braced{0,1,...,q-1}$, if $\barr$ contains a length-$j$ jump sub-sequence then
associate with $\barr$ the unique \emph{lexicographically-first appearance} of a length-$j$ jump sub-sequence.
The number of $\barr$'s that contain a length-$j$ jump sub-sequence 
but not a length-$(j+1)$ jump sub-sequence can be
obtained by choosing the lexicographically first length-$j$ jump sub-sequence in
$\binom{\ell}{j}\cdot (k/q)^j$ ways and multiplying it by
the number of ways of filling the remaining $\ell-j$ positions, which is $(k-k/q)^{\ell-j}$.
We thus have
\begin{align*}
	p'_{\ell,k,q} \;\le\; p''_{\ell,k,q} \;&=\; \frac{1}{k^\ell} \sum_{j=0}^{q-1} \binom{\ell}{j}\cdot (k/q)^j \cdot (k - k/q)^{\ell-j}
	\\ 
	\;&=\; (1-1/q)^\ell \sum_{j=0}^{q-1}  \binom{\ell}{j} (q-1)^{-j}	
	\\
		\;&\le\; (1-1/q)^\ell \sum_{j=0}^{q-1} (2\ell /q)^j	
	\\
		\;&\le\; e^{-\ell/q} (2\ell /q)^q.
\end{align*}
This proves~\eqref{eqn: Nkq-subsequence probability} (with $\gamma = 1/e$), completing the proof.
\end{proof}

For $q=k$, the bound in Theorem~\ref{thm: Nkq permutation selector construction} matches
the bound from Theorem~\ref{thm: main result} and the best bound for strong \kNpair{k}{N}-selectors.
For $q=1$, it also matches the $O(k\log N)$ bound for weak \kNpair{k}{N}-selectors.
We are not sure about the optimum size of \kqNtriple{k}{q}{N}-permutation selectors for intermediate values of $q$. 
The case when $q = \sqrt{k}$ is particularly interesting. From the bound in~\cite{DeBonis_etal_selectors},
\kqNtriple{k}{\sqrt{k}}{N}-selectors have size $O(k\log N)$. 
Can some probabilistic construction produce a \kqNtriple{k}{\sqrt{k}}{N}-permutation selector
of size $m = \tildeO(k)$? For such $m$, a random sequence of coupons from $\braced{0,1,...,k-1}$
can be thought of as a ``noisy'' permutation. We would need to show that the
probability that this random sequence does not contain an increasing sub-sequence of length $\sqrt{k}$ is
exponentially small. This question is closely related to Ulam's Problem
about the distribution of the longest increasing subsequence (LIS) in a random permutation.
Ulam's Problem has been extensively studied, and it is known that for permutations of $0,1,...,k-1$
the expected length of LIS is $2\sqrt{k}+ O(k^{1/6})$. 
(See~\cite{romik_lis_book_2014}, for example.) To our knowledge, however, the published concentration bounds are 
not sufficient for refining the probabilistic construction in the
proof of Theorem~\ref{thm: Nkq permutation selector construction} to yield a better bound.



\end{document}